\documentclass[12pt,a4paper,reqno]{amsart}
\usepackage[hmargin=2.5cm,bmargin=2.5cm,tmargin=3cm]{geometry}
\usepackage{enumerate}
\usepackage{amsmath,amsthm,amssymb,amsfonts,latexsym}
\usepackage[gen]{eurosym}
\usepackage[german,english]{babel}

\usepackage{latexsym}

\usepackage{tikz}

\newtheorem{theorem}{Theorem}
\newtheorem{lemma}[theorem]{Lemma}

\newtheorem{proposition}[theorem]{Proposition}

\newtheorem{remark}[theorem]{Remark}

\newcommand{\be}{\begin{equation}}
\newcommand{\ee}{\end{equation}}
\newcommand{\bea}{\begin{eqnarray*}}
	\newcommand{\eea}{\end{eqnarray*}}
\newcommand{\beq}{\begin{eqnarray}}
\newcommand{\eeq}{\end{eqnarray}}

\usepackage{todonotes}

\title[Comparing the spectrum of Schrödinger operators on quantum graphs]{Comparing the spectrum of Schrödinger operators on quantum graphs} 

\subjclass[2010]{}


\author[P.~Bifulco]{Patrizio Bifulco}
\author[J.~Kerner]{Joachim Kerner}

\address{Lehrgebiet Analysis, Fakult\"at Mathematik und Informatik, Fern\-Universit\"at in Hagen, D-58084 Hagen, Germany}
\email{patrizio.bifulco@fernuni-hagen.de}

\address{Lehrgebiet Analysis, Fakult\"at Mathematik und Informatik, Fern\-Universit\"at in Hagen, D-58084 Hagen, Germany}
\email{joachim.kerner@fernuni-hagen.de}

\date{\today}

\thanks{
}

\begin{document}

	\begin{abstract} We study Schrödinger operators on compact finite metric graphs subject to $\delta$-coupling and standard boundary conditions. We compare the $n$-th eigenvalues of those self-adjoint realizations and derive an asymptotic result for the mean value of deviations. By doing this, we generalize recent results from~\cite{RWY} obtained for domains in $\mathbb{R}^2$ to the setting of quantum graphs. This also leads to a generalization of related results previously and independently obtained in~\cite{GiladSofer2022} and~\cite{BSSS} for metric graphs. In addition, based on our main result, we introduce some notions of \emph{circumference} for a (quantum) graph which might prove useful in the future.
	\end{abstract}
	
\maketitle

	\section{Introduction}
	
	Our paper is primarily motivated by a recent work of Rudnick, Wigman, and Yesha in which they study the difference of Robin and Neumann eigenvalues for the Laplacian on domains in $\mathbb{R}^2$. Looking at the difference of the corresponding $n$-th eigenvalues, $d_n(\sigma):=\lambda_n(\sigma)-\lambda_n(0)$,  they derive a limit for the mean value of gaps. More explicitly, for bounded domains $\Omega \subset \mathbb{R}^2$ with piecewise smooth boundary, they show that the limit of the arithmetic mean of $d_n(\sigma)$ equals
	\begin{equation*}
	\frac{2|\partial \Omega|}{|\Omega|}\cdot \sigma
	\end{equation*}
with $\sigma > 0$ denoting the Robin parameter. In addition to that, they also study upper and lower bounds on $d_n(\sigma)$. It is interesting to note that similar questions were already studied in~\cite{RR} for star graphs with Dirichlet ends and, in particular, in~\cite{GiladSofer2022} and~\cite{BSSS} for general metric graphs, respectively. However, in this paper we shall follow the strategy of~\cite{RWY} to derive an explicit limit value for the arithmetic mean of $d_n(\sigma)$ on general finite compact, connected quantum graphs. It is important to note that, using results of~\cite{BER} and combining them with some methods of~\cite{BHJ}, we are able to allow for smooth bounded potentials supported on the edges of the graph. A key ingredient in the proof of our main result is a so-called \emph{local Weyl law} which has been established for graphs (with standard boundary conditions) already in~\cite{BHJ} but without potentials. In this way, we are able to generalize some results from~\cite{BHJ}, too. At this point we should note that, for certain boundary conditions and without potentials, a local Weyl law is also derived in~\cite{GiladSofer2022} and~\cite{BSSS} using different methods. 

 Furthermore, based on results of~\cite{KS}, we are able to prove a uniform upper bound on $d_n(\sigma)$; by providing a counterexample, we show that a uniform lower bound cannot hold for arbitrary graphs which was already noted in~\cite{RR} (see, for instance, Proposition 1.6) in their setting (for a more detailed discussion on upper bounds and the distribution of $d_n(\sigma)$ we refer to~\cite{BSSS}). Finally, based on our main result, we propose some notions of \emph{circumference} for a (quantum) graph which might prove interesting in the future. For example, what impact has the circumference on the spectral properties of the underlying (quantum) graph?

\section{The setting}

	We study compact, connected finite quantum graphs $\Gamma=\Gamma(\mathcal{E},\mathcal{V})$ with (finite) edge set $\mathcal{E}$ and vertex set $\mathcal{V}$. To each edge $e \in \mathcal{E}$ one associates an interval $(0,l_e)$ with $0 < l_e < \infty$ denoting the \emph{edge length}. The number of edges connected to a vertex $v \in \mathcal{V}$ is called the \textit{degree} and shall be denoted by $\deg_v \in \mathbb{N}$. The relevant Hilbert space is given by 
	\begin{equation}\label{HilbertSpace}
	L^2(\Gamma)=\bigoplus_{e \in \mathcal{E}}L^2(0,l_e)\ .
	\end{equation}
	We write $\mathcal{L}=\sum_{e \in \mathcal{E}}l_e$ for the total length of the graph $\Gamma$. On $L^2(\Gamma)$ we then introduce the Schrödinger operator 
	\begin{equation}\label{SchrödingerOperator}
	-\Delta+V
	\end{equation}
	which acts, on suitable functions $f=(f_e)_{e \in \mathcal{E}} \in L^2(\Gamma)$, via 
	\begin{equation*}
	[\left(-\Delta+V)f\right]_e(x)=-f_e^{\prime \prime}(x)+v_e(x)f_e(x)\ , \quad x \in (0,l_e)\ .
	\end{equation*}
	In this paper we shall assume that the potential is real-valued with $v_e \in (C^{\infty}\cap L^{\infty})(0,l_e)$ for all $e \in \mathcal{E}$. This assumption guarantees that the operator $-\Delta+V$ is self-adjoint on a domain on which $-\Delta$ is self-adjoint. By $V_{-}(x):=-\min\{0,V(x)\}$ we denote the negative part of the potential $V$.
	
	We are particularly interested in two domains of self-adjointness of $-\Delta$: The well-known Kirchoff-Neumann (or standard) realization is obtained by defining $-\Delta$ on the domain
	\begin{equation}
	\mathcal{D}_{0}:=\bigg\{f\in H^1(\Gamma): f_e \in H^2(0,l_e) \ \text{and}\  \sum_{e \sim v}\partial_v f_e(v)=0  \bigg\}
	\end{equation}
	where $H^1(\Gamma)$ consists of all functions $f \in L^2(\Gamma)$ for which $f_e \in H^1(0,l_e)$, $e \in \mathcal{E}$, and which are continuous across vertices. Note that we write $f_e(v)$ for the boundary value of $f_e \in H^1(0,l_e)$ at either $x=0$ or $x=l_e$, depending on the given vertex $v \in \mathcal{V}$. We also sometimes write $f(v)$ for functions $f \in H^1(\Gamma)$, referring to the value of one $f_e(v)$ at the given vertex. In addition, $\partial_v f_e(v)$ refers to the \emph{inward} derivative of the component $f_e(x)$ at either $x=0$ or $x=l_e$ (note that this leads to a minus sign whenever $x=l_e$).
	
	The other self-adjoint realization of $-\Delta$ we are interested in  comes with the domain 
	\begin{equation}
	\mathcal{D}_{\sigma}:=\bigg\{f\in H^1(\Gamma): f_e \in H^2(0,l_e) \ \text{and}\  \sum_{e \sim v}\partial_v f_e(v)=\sigma_v f_e(v)  \bigg\}
	\end{equation}
	where $\sigma\equiv (\sigma_v)_{v \in \mathcal{V}} \in \mathbb{R}^{|\mathcal{V}|}$. Note that this self-adoint realization is well-known under the name $\delta$-coupling condition. Furthermore, note that the standard Kirchoff-Neumann realization is obtained through choosing $\sigma= (0,\dots,0)$. 
	
	Consequently, we are dealing with the two self-adjoint (Schrödinger) operators $H^V_0=(-\Delta+V,\mathcal{D}_{0})$ and $H^V_{\sigma}=(-\Delta+V,\mathcal{D}_{\sigma})$. Since the underlying quantum graph is finite and compact, both operators have purely discrete spectrum. We shall denote the eigenvalues of $H^V_0$ as $\lambda^V_0(0) \leq \lambda^V_1(0) \leq \dots$ and the eigenvalues of $H^V_{\sigma}$ as $\lambda^V_0(\sigma) \leq \lambda^V_1(\sigma) \leq \dots$, counting them with multiplicity in both cases. Also, $f^{0,V}_n$ and $f^{\sigma,V}_n$ refer to corresponding orthogonal and normalized eigenfunctions.
		
	Motivated by~\cite{RWY} and~\cite{KS}, the quantities we are interested in are defined by
\begin{equation}\begin{split}\label{DefinitionDifference}
d^V_n(\sigma)&:=\lambda^V_n(\sigma)-\lambda^V_n(0)\ , \\
\hat{d}^V_n(\sigma)&:=\lambda^V_n(\sigma)-\lambda^{V=0}_n(0)\ , \quad n \in \mathbb{N}\ .
\end{split}
\end{equation}

	We now formulate our main result which is a generalized version of~\cite[Theorem~1.1]{RWY} for quantum graphs. We note that a special version of this theorem was previously and independently obtained in~\cite[Theorem 3.3]{GiladSofer2022} and~\cite[Theorem 1.3]{BSSS}, respectively. 
	
	\begin{theorem}\label{MainResult} Let $H^V_0$ and $H^V_{\sigma}$ be the two self-adjoint Schrödinger operators on a finite compact, connected quantum graph $\Gamma$ with properties described above. Then one has
		\begin{equation}\label{RelationI}
		\lim_{N \rightarrow \infty}\frac{1}{N}\sum_{n=1}^{N} d^V_n(\sigma)=\frac{2}{\mathcal{L}}\cdot \sum_{v \in \mathcal{V}} \frac{\sigma_v}{\deg_v}\ .
		\end{equation}
	Furthermore, comparing $H^{V=0}_0$ with $H^V_{\sigma}$ one obtains
		\begin{equation}\label{RelationII}
	\lim_{N \rightarrow \infty}\frac{1}{N}\sum_{n=1}^{N} \hat{d}^V_n(\sigma)=\frac{2}{\mathcal{L}}\left(\sum_{v \in \mathcal{V}} \frac{\sigma_v}{\deg_v}+\frac{1}{2}\sum_{e \in \mathcal{E}}\int_{0}^{l_e}v_e(x)\mathrm{d} x\right)\ .
	\end{equation}
	\end{theorem}

 In \cite[Theorem 1.3]{RR}, the mean-value~\eqref{RelationI} has been investigated for a compact star graph but with Dirichlet boundary conditions (!)~at the edge ends and a $\delta$-coupling condition at the central vertex. Applying our main result to a star graph with $|\mathcal{E}| \in \mathbb{N}$ leads, $\sigma_v=0$ for all vertices forming an edge end and $\sigma_v=\sigma$ for the central vertex, we find that 
	\begin{equation*}\label{StarGraph}
	\lim_{N \rightarrow \infty}\frac{1}{N}\sum_{n=1}^{N} d^V_n(\sigma)=\frac{2\sigma}{|\mathcal{E}|\mathcal{L}}\ .
	\end{equation*}

	\section{Proof of Theorem~\ref{MainResult} and auxiliary results}
	In this section our goal is to prove Theorem~\ref{MainResult}, following the strategy developed in~\cite{RWY}. We also rely on results and methods established in \cite{BER} as well as \cite{BHJ}. In a first result we express $d^V_n(\sigma)$ and $\hat{d}^V_n(\sigma)$ in a new and suitable way, relating them to the boundary values of eigenfunctions.
	\begin{lemma}\label{LemmaI} One has 
		\begin{equation*}\label{Relation}
		d^V_n(\sigma)=\int_{0}^{1}\sum_{v \in \mathcal{V}}\sigma_v|f^{\tau\sigma,V}_n(v)|^2\ \mathrm{d}\tau
		\end{equation*}
		and 
			\begin{equation*}\label{RelationIII}
		\hat{d}^V_n(\sigma)=\int_{0}^{1}\left(\sum_{e \in \mathcal{E}}\int_{0}^{l_e}v_e(x)|f^{\tau\sigma,\tau V}_n(x)|^2\mathrm{d} x+\sum_{v \in \mathcal{V}}\sigma_v|f^{\tau\sigma,\tau V}_n(v)|^2\ \right) \mathrm{d}\tau
		\end{equation*}
		where $\tau\sigma=(\tau\sigma_1,...,\tau\sigma_v)$, $\tau V=(\tau v_e)_{e \in \mathcal{E}}$ and $\tau\in [0,1]$.
	\end{lemma}
	\begin{proof} The statement is a suitable adaptation of~\cite[Lemma~3.1]{RWY} which again is based on the general theory developed in~\cite{Kat66}. More explicitly, in the language of Kato, we are dealing with families of self-adjoint forms of type (B). Note that the quadratic form associated with $H^V_{\sigma}$ is given by 
		\begin{equation}\label{QFrom}
		q_{\sigma,V}[f]:=\sum_{e\in \mathcal{E}}\int_{0}^{l_e}\left(|f^{\prime}_e(x)|^2+v_e(x)\vert f_e(x) \vert^2\right)\ \mathrm{d}x+\sum_{v \in \mathcal{V}}\sigma_v|f_e(v)|^2\ ,
		\end{equation}
		defined on $H^1(\Gamma)$. Consequently,~\cite{Kat66} (or the \emph{Feynman-Hellman theorem}) allows to justify the relations 
		\begin{equation*}
	\frac{\mathrm{d} \lambda^V_n(\tau\sigma)}{\mathrm{d} \tau}=\sum_{v\in \mathcal{V}}\sigma_v |f^{\tau\sigma}_n(v)|^2
		\end{equation*}
		and 
			\begin{equation*}
		\frac{\mathrm{d} \lambda^{\tau V}_n(\tau\sigma)}{\mathrm{d} \tau}=\sum_{e \in \mathcal{E}}\int_{0}^{l_e}v_e(x)|f^{\tau\sigma,\tau V}_n(x)|^2\mathrm{d} x+\sum_{v\in \mathcal{V}}\sigma_v |f^{\tau\sigma,\tau V}_n(v)|^2
		\end{equation*}
		for almost all $\tau \in [0,1]$. From this the claimed relation readily follows by integration.
	\end{proof}
We also need the following auxiliary result which is also interesting in its own right.

\begin{lemma}\label{Dominated}
	Consider a graph $\Gamma$ and a Schrödinger operator $H^V_{\sigma}$. Then, there exists a constant $C > 0$ such that one has
	\[
 \| f_n^{\tau\sigma,V}\|_{L^{\infty}(\Gamma)} \leq C\ ,
	\]
	and 
	\[
	\| f_n^{\tau\sigma,\tau V}\|_{L^{\infty}(\Gamma)} \leq C\ ,
	\]
	for all $n \in \mathbb{N}$ and all $\tau \in [0,1]$.
\end{lemma}
\begin{proof} We first assume that $V=0$ and outline a generalization to the case with potential at the end. Any (not necessarily normalized) eigenfunction $g_n^{\tau\sigma}$ to a positive eigenvalue $\lambda^{V=0}_n(\sigma)  > 0$ is a collection of $|\mathcal{E}|$ components 
	\begin{equation}\label{EqProof}
	(g_n^{\tau\sigma})_e(x)=a_{n,e}\mathrm{e}^{ik_nx}+b_{n,e}\mathrm{e}^{-ik_nx}
	\end{equation}
	with certain (complex) coefficients $a_{n,e},b_{n,e}$ that are given as entries of a (unit) eigenvector of a matrix that incorporates the boundary conditions in the vertices (see, for instance,~\cite{Berkolaiko}). Hence, one has
	$$\sum_{e \in \mathcal{E}}(|a_{n,e}|^2+|b_{n,e}|^2)=1\ .$$
	In a next step one chooses a factor $\mu_n(\tau\sigma)$ such that $\mu_n(\tau \sigma)g_n^{\tau\sigma}$ has $L^2(\Gamma)$-norm one. A calculation now shows that there exist constants $c_1,c_2 > 0$ such that, for $n$ large enough, 
	\begin{equation*}
	c_1 \leq \mu_n(\tau\sigma) \leq c_2
	\end{equation*}
	holds for all $\tau \in [0,1]$. This then implies that each normalized eigenfunction is uniformly bounded for all $n$ large enough. Finally, for the remaining small $n$, one employs the well-known inequality (see, e.g.,~\cite[Lemma 3]{KS})
	$$\|f\|^2_{L^{\infty}(0,l)} \leq \varepsilon\|f^{\prime}\|^2_{L^2(0,l)}+\frac{2}{\varepsilon}\|f\|^2_{L^2(0,l)}\ , $$
	which holds for $f \in H^1(0,l)$ and $0 < \varepsilon < l$.  With this we obtain, for $0 < \varepsilon < \min_{e \in \mathcal{E}}l_e$,
	\begin{equation*}\begin{split}
	\|\mu_n(\tau\sigma)g_n^{\tau\sigma}\|^2_{L^{\infty}(\Gamma)} &\leq \varepsilon\|(\mu_n(\tau\sigma)g_n^{\tau\sigma})^{\prime}\|^2_{L^2(\Gamma)}+\frac{2}{\varepsilon}\ , \\
	&\leq \varepsilon \lambda^{V=0}_n(\infty)+\frac{2}{\varepsilon}
	\end{split}
	\end{equation*}
	where $\lambda^{V=0}_n(\infty)$ is the $n$-th Dirichlet eigenvalue (meaning one demands Dirichlet boundary conditions in all vertices). This implies the statement.
	
	For non-vanishing potential $V$, the arguments are similar. In~\eqref{EqProof}, $\mathrm{e}^{ik_nx}$ and $\mathrm{e}^{-ik_nx}$ have to be replaced by suitable fundamental solutions which, for large $k_n$ and as stated in \cite[Theorem~3.2]{RalfRueckriemen}, are close to plane waves and hence the arguments from above can be repeated (see also \cite[Lemma~5.1]{BER} and \cite[Remark~5.2]{BER}).
	\end{proof}
	
	In a next step we derive a so-called \emph{local Weyl law} as discussed, for example, in \cite[Theorem 4.1]{BHJ}. Note that, since we also allow for potentials, we are able to generalize some results of \cite{BHJ}.
	\begin{proposition}[Local Weyl law]\label{prop:localweyllaw}
		One has
		\[
		\lim_{N \rightarrow \infty} \frac{1}{N} \sum_{n=1}^N |(f^{\sigma,V}_n)_e(x)|^2 = \frac{2}{\mathcal{L}\deg_x}\ ,
		\]
	where $\deg_x =2$ if $x \in (0,l_e)$ and $\deg_x=\deg_v$ if $x=v$ with $v\in \mathcal{V}$.
	\end{proposition}
	\begin{proof}
	Introducing the integral kernel of the operator $\mathrm{e}^{-H_{\sigma}t}$ on $\Gamma$ (the so-called \emph{heat kernel}), for every $e \in \mathcal{E}$ and $x \in [0,l_e]$, \cite[Proposition~8.2]{BER} implies 
	\[
	p_{e,e}^{H^V_{\sigma}}(t;x,x) \sim \frac{1}{\sqrt{4\pi t}} \frac{2}{\deg_x} \:\: \text{as} \:\: t \rightarrow 0^+.
	\]
	On the other hand, one has the following expansion 
	\begin{equation*}
	\sum_{n=1}^{\infty}\mathrm{e}^{-\lambda^V_n(\sigma) t}|(f^{\sigma,V}_n)_e(x)|^2=p_{e,e}^{H^V_{\sigma}}(t;x,x)\ .
	\end{equation*}
	Now, as in~\cite{BHJ}, we can employ Karamata's Tauberian theorem to conclude
	\begin{align}\label{agn:vorstufeocalweyllaw}
	\sum_{\lambda^V_n(\sigma) \leq \lambda} |(f^{\sigma,V}_n)_e(x)|^2 \sim  \frac{2}{\pi \deg_x}\lambda^\frac{1}{2} \:\: \text{as} \:\: \lambda \rightarrow \infty\ .
	\end{align}
	In a last step we can employ Weyl's law and Lemma~\ref{Dominated} to  replace the condition $\lambda^V_n(\sigma) \leq \lambda$ by $n \leq N$. More explicitly, one has $N \sim \frac{\mathcal{L}}{\pi} \lambda^\frac{1}{2}$ or equivalently $(\frac{\pi N}{\mathcal{L}})^2 \sim \lambda$. This yields
	\[
	\sum_{n=1}^N |(f^{\sigma,V}_n)_e(x)|^2 \sim \frac{2N}{\mathcal{L}\deg_x} \:\: \Longleftrightarrow \:\: \frac{1}{N} \sum_{n=1}^N |(f^{\sigma,V}_n)_e(x)|^2 \sim \frac{2}{\mathcal{L}\deg_x}\ ,
	\]
	which shows the claim.
\end{proof}

\begin{remark} Using results of~\cite{BER}, the local Weyl law can also be formulated for more general (local) boundary conditions. More explicitly, assuming self-adjoint boundary conditions in a vertex $v \in \mathcal{V}$ described by a pair of matrices $\left(P_v,L_v\right)$ (with the well-known properties as described, for example, in~\cite{BER}), one obtains 
	\[
\lim_{N \rightarrow \infty} \frac{1}{N} \sum_{n=1}^N |(f^{\sigma,V}_n)_e(v)|^2 = \frac{1+(S_v(\infty))_{e,e}}{\mathcal{L}}\ ,
\]
where $S_v(\infty):=\mathrm{Id}-P_v$, see \cite[Theorem~5.2]{BER}. Evaluating the left-hand side of the previous formula at some $x \in (0,l_e)$, one obtains the same result as in Proposition~\ref{prop:localweyllaw}. 
\end{remark}

We are now able to give the proof of our main result. 
	
	\begin{proof}[Proof of Theorem~\ref{MainResult}] We follow the strategy of~\cite{RWY}.
		Lemma~\ref{LemmaI} gives
		\begin{equation*}\begin{split}
		\frac{1}{N}\sum_{n=1}^{N} d^V_n(\sigma)&=	\frac{1}{N}\sum_{n=1}^{N}\int_{0}^{1}\sum_{v \in \mathcal{V}}\sigma_v|f^{\tau\sigma,V}_n(v)|^2\ \mathrm{d}\tau\\
		&=\int_{0}^{1}\sum_{v \in \mathcal{V}}\sigma_v \left(\frac{1}{N}\sum_{n=1}^{N}|f^{\tau\sigma,V}_n(v)|^2\right)\ \mathrm{d}\tau\ .
		\end{split}
		\end{equation*}
		
		Due to Proposition~\ref{prop:localweyllaw} and Lemma~\ref{Dominated}, which allows to make use of dominated convergence, we obtain
		\begin{equation*}\begin{split}
		\lim_{N \rightarrow \infty} \frac{1}{N}\sum_{n=1}^{N} d^V_n(\sigma) &=\int_{0}^{1}\sum_{v \in \mathcal{V}}\sigma_v \lim_{N \rightarrow \infty}\left(\frac{1}{N}\sum_{n=1}^{N}|f^{\tau\sigma,V}_n(v)|^2\right)\ \mathrm{d}\tau \\
		&=\frac{2}{\mathcal{L}}\cdot \sum_{v \in \mathcal{V}} \frac{\sigma_v}{\deg_v}\ . 
		\end{split}
		\end{equation*}
		In the same fashion, we conclude that 
		\begin{equation*}\begin{split}
		\lim_{N \rightarrow \infty} \frac{1}{N}\sum_{n=1}^{N} \hat{d}^V_n(\sigma)=&\int_{0}^{1}\left(\sum_{e \in \mathcal{E}}\int_{0}^{l_e}v_e(x)\lim_{N \rightarrow \infty}\left(\frac{1}{N}\sum_{n=1}^{N}|f^{\tau\sigma,\tau V}_n(x)|^2\right)\mathrm{d} x \right) \mathrm{d} \tau \\		
		&+\int_{0}^{1}\sum_{v \in \mathcal{V}}\sigma_v \lim_{N \rightarrow \infty}\left(\frac{1}{N}\sum_{n=1}^{N}|f^{\tau\sigma, \tau V}_n(v)|^2\right)\ \mathrm{d}\tau \\
		&=\frac{2}{\mathcal{L}}\left(\sum_{v \in \mathcal{V}} \frac{\sigma_v}{\deg_v}+\frac{1}{2}\sum_{e \in \mathcal{E}}\int_{0}^{l_e}v_e(x)\mathrm{d} x \right)\ .
		\end{split}
		\end{equation*}
	\end{proof}

	\begin{remark} Comparing the result of Theorem~\ref{MainResult} with the results of~\cite[Theorem 1.1]{RWY} allows to introduce some possible notions of \textbf{circumference} for a (compact) quantum graph. Namely, by setting
		\begin{equation}\label{CircumferenceI}
		\mathcal{C}(\Gamma):=\sum_{v \in \mathcal{V}}\frac{1}{\deg_v}
		\end{equation}
		one obtains a combinatorial measure thereof. This definition has an interesting interpretation in terms of ``surface tension'': a water molecule in the middle of the sea (implying its not part of the surface) is surrounded by many other molecules and in this sense its degree so large that it doesn't contribute to the surface as expected. On the other hand, a water molecule at the surface has less neighbours and hence it contributes more to the total surface. In the same spirit one could justify that the boundary of a graph is formed by vertices of degree one. 
		
		Since $\mathcal{C}(\Gamma)$ is, for example, not immune to the introduction of ``dummy vertices'' (meaning extra vertices placed along some edges $e \in \mathcal{E}$ with Kirchoff-Neumann boundary conditions) one might also define an \textbf{effective circumference} via 
		\begin{equation}\label{CircumferenceII}
		\mathcal{C}_{\sigma}(\Gamma):=\sum_{v \in \mathcal{V}}\frac{\sigma_v}{\deg_v}\ .
		\end{equation}
		Since Dirichlet vertices $v\in \mathcal{V}$ formally correspond to vertices with $\sigma_v=\infty$, it seems plausible to conclude that the (effective) boundary of a graph is given by the union of all Dirichlet vertices whenever existing. 
		
		Finally, it might also be interesting to introduce a certain level $\varepsilon > 0$ to say that the ($\varepsilon$-)boundary of a (quantum) graph is made up of all vertices for which $\frac{\sigma_v}{\deg_v} > \varepsilon$. 
		\end{remark}
		\section{Lower and upper bounds on $d^V_n(\sigma)$ and $\hat{d}^V_n(\sigma)$}
	In this section we want to discuss upper and lower bounds on the spectral distances $d^V_n(\sigma)$ and $\hat{d}^V_n(\sigma)$ in the limit $n \rightarrow \infty$. For star graphs with Dirichlet boundary conditions at the edge ends, this was already investigated in \cite[Proposition 1.6]{RR} and it was found that a uniform upper bound exists while a uniform positive lower bound is not available. As described in \cite{RWY}, on domains in $\mathbb{R}^2$ the situation is quite different: For general domains with smooth boundary one is able to derive an upper bound which is of order $n^{\frac{1}{3}}$; however, as shown in \cite[Section 8]{RWY}, on rectangles an uniform upper bound can be derived. In addition to that, for star-shaped domains with a smooth boundary, a uniform positive lower bound can be derived as well, see \cite[Theorem~1.3]{RWY}.
	\begin{proposition}[Upper bound]
		There exists a constant $C=C(\Gamma,V,\sigma) > 0$ such that
		\[
		 |d^V_n(\sigma)|\leq C \quad \text{and} \quad |\hat{d}^V_n(\sigma)| \leq C
		\]
		for every $n \in \mathbb{N}$.
	\end{proposition}
	\begin{proof} The statement follows from a suitable application of \cite[Theorem~4]{KS} in which a uniform bound on $|\hat{d}^V_n(\sigma)|$ has been derived. Regarding $d^V_n(\sigma)$, the triangle inequality and using \cite[Theorem~4]{KS} twice gives
		\[
		|d^V_n(\sigma)| \leq \vert \lambda^V_n(\sigma) - \lambda^{V=0}_n(0) \vert + \vert \lambda^{V=0}_n(0)-\lambda^V_n(0) \vert \leq C_1(\Gamma,V,\sigma) + C_2(\Gamma,V)
		\]
		for every $n \in \mathbb{N}$.
	\end{proof}
Based on the results of~\cite{RR} as mentioned above, a lower bound for $d^V_n(\sigma)$ cannot be expected to hold. To see this in our setting, we construct a rather simple counterexample.

\begin{proposition}\label{PropStarGraph} Let $\Gamma = \Gamma(\mathcal{E}, \mathcal{V})$ denote a star graph consisting of two edges $\mathcal{E} = \{e_1,e_2 \}$ with same length $l_{e_1}=l_{e_2}=:l > 0$ and three vertices $\mathcal{V} = \{v_c,v_1,v_2\}$ where $v_c \in \mathcal{V}$ denotes the \emph{central vertex} and $v_1,v_2 \in \mathcal{V}$ the two outer vertices (see Figure~1). Furthermore, set $\sigma=(\sigma,0,0)$ with $\sigma > 0$ and $v_e\equiv 0$ for all $e \in \mathcal{E}$. Then,
	\begin{equation*}
	d^{V=0}_{2n}(\sigma)=0\ , \quad \forall n \in \mathbb{N}\ .
	\end{equation*}
\end{proposition}
\begin{figure}[h]
	\begin{tikzpicture}[scale=0.60]
	\tikzset{enclosed/.style={draw, circle, inner sep=0pt, minimum size=.15cm, fill=gray}, every loop/.style={}}
	\node[enclosed, label={below: $v_1$}] (V1) at (0,4) {};
	\node[enclosed, label={below: $v_2$}] (V2) at (10,4) {};
	\node[enclosed, label={below: $v_c$}] (C) at (5,4) {};
	
	\draw (V1) edge [bend left=25] node[above] {$e_1$} (C) node[midway, above] (edge1) {};
	\draw (C) edge [bend right=25] node[below] {$e_2$} (V2) node[midway, above] (edge2) {};
	\end{tikzpicture}
	\vspace{-1.5cm}
	\caption{$2$-star with central vertex $v_c$ and boundary vertices $v_1,v_2$}
\end{figure}
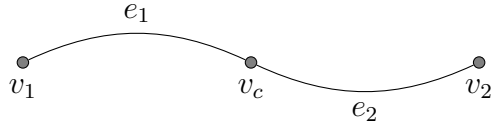

\begin{proof}[Proof of Proposition~\ref{PropStarGraph}] The statement follows from the simple observation that all even eigenfunctions of the operator $(-\Delta,\mathcal{D}_0)$ on the considered star graph are also eigenfunctions of the operator $(-\Delta,\mathcal{D}_\sigma)$ to the same eigenvalue since each such eigenfunction is zero at the central vertex $v_c$. Furthermore, the order of the eigenfunctions agree and this shows that $d_{2n}(\sigma)=0$ for all $n \in \mathbb{N}$. 
	
	To see this, note that each component of an eigenfunction corresponding to the eigenvalue $\lambda_n(\sigma)=k^2_n(\sigma)$ can be written as
	\begin{equation*}
	(f^{\sigma}_n)_e(x)=a_e \cos(k_n(\sigma) x)\ ,
	\end{equation*}
	where $x=0$ corresponds to the outer vertex in each case. Now, the vertex conditions at the central vertex $v_c$ either imply 
	\begin{equation}\label{EqX}
	\cos(k_n(\sigma) l)=0
	\end{equation}
	in combination with $a_1=-a_2$ (this gives the anti-symmetric eigenfunctions) or one has $a_1=a_2$ in combination with
	\begin{equation}\label{EqY}
	\tan(k_n(\sigma)l)=\frac{\sigma}{2k_n(\sigma)}\ ,
	\end{equation}
	yielding the symmetric eigenfunctions. Comparing the solutions of \eqref{EqX} and \eqref{EqY} one concludes that all even eigenvalues are squares of solutions of~\eqref{EqX} which also equal $\lambda_{2n}(0)$. Since all eigenvalues are simple, this gives the statement.
\end{proof}
\appendix
\section{An alternative to Lemma~\ref{Dominated}}
In this appendix we establish a version of Lemma~\ref{Dominated} which is sufficient to prove our main result by allowing to use dominated convergence in the proof of Theorem~\ref{MainResult}; it forms an analogue to \cite[Lemma~4.1]{RWY}. More explicitly, we show the following statement.

\begin{lemma}\label{DominatedI}
	Let $\Gamma$ be a graph with properties described above and $\hat{\sigma},\hat{V}_{-},\hat{V}_{+} \geq 0$ some constants. Then, there exists a constant $C=C(\Gamma,\hat{\sigma},\hat{V}_{-},\hat{V}_{+} ) > 0$ such that for every Schrödinger operator $H^V_{\sigma}$ with $\sigma = (\sigma_1,\dots,\sigma_{\vert \mathcal{V} \vert}) \in [-\hat{\sigma},\infty)^{\vert \mathcal{V} \vert}$ and potential $V$ such that $\|V_{-}\|_{L^{\infty}(\Gamma)} \leq \hat{V}_{-}$, $\|V\|_{L^{\infty}(\Gamma)} \leq \hat{V}_+$ one has, for every $N \in \mathbb{N}$,
	\[
	\frac{1}{N} \sum_{n=1}^N \sum_{v \in \mathcal{V}} \vert f_n^{\sigma,V}(v)\vert^2 \leq C \ 
	\]
	and, for for every $N \in \mathbb{N}$ and every $x \in (0,l_e)$,
	\[
	\frac{1}{N} \sum_{n=1}^N \vert (f_n^{\sigma,V})_e(x)\vert^2 \leq C\ .
	\]
\end{lemma}

 The proof of Lemma~\ref{DominatedI} uses a domination argument for the heat kernel $p^{H^V_\sigma}(t;x,y)$ of the semigroup $(e^{-tH^V_\sigma})_{t \geq 0}$ based on an abstract argument by Ouhabaz (see~\cite{Ouh05}): Indeed, consider the form $q_{0,V=0}$ associated with the Laplacian and standard Kirchhoff-Neumann conditions and the form $q_{\sigma,V}+(\hat{\sigma}+\hat{V}_-)\cdot \|\cdot\|^2_{L^2(\Gamma)}$ associated with operator $H^V_{\sigma}+\hat{\sigma}+\hat{V}_-$ and $\delta$-type vertex conditions (compare with~\eqref{QFrom}). Then, for non-negative functions in the form domain $H^1(\Gamma)$, one has $q_{0,V=0}[f,g] \leq q_{\sigma,V}[f,g]+(\hat{\sigma}+\hat{V}_-)\langle f,g \rangle_{L^2(\Gamma)}$ for all non-negative $f,g \in H^1(\Gamma)$. \cite[Theorem 2.21]{Ouh05} then implies 
\[
\vert \mathrm{e}^{-t H^{V}_\sigma} f \vert \leq  \mathrm{e}^{(\hat{\sigma}+\hat{V}_-)t}\cdot \mathrm{e}^{-t H^{V=0}_0} \vert f \vert 
\] 
for every $f \in L^2(\Gamma)$ and $t \geq 0$. This eventually yields

\begin{align}\label{agn:15}
p^{H^{V}_\sigma}(t;x,y) \leq \mathrm{e}^{(\hat{\sigma}+\hat{V}_-)t}p^{H^{V=0}_0}(t;x,y) 
\end{align}
for every $x,y \in \Gamma$ and $t \geq 0$. 

\begin{proof}
	Using \cite[Proposition~8.1]{BER} as well as \eqref{agn:15}, we conclude that for each vertex $v \in \mathcal{V}$ one has
	\begin{align}\label{agn:wurzelabschätzungheatkernel}
	p^{H^{V}_\sigma}_{e,e}(t;v,v) \leq \frac{C_1(\Gamma,\hat{\sigma},\hat{V}_-)}{\sqrt{t}}
	\end{align} 
	for some constant $C_1(\Gamma,\hat{\sigma},\hat{V}_-) > 0$ (which does not depend on the considered $\sigma$ or $V$) and $t$ small enough. 
	
The eigenvalues $\lambda^V_n(\sigma)$ of $H^V_\sigma$ are increasing and clearly dominated by the corresponding Dirichlet eigenvalues $\lambda^V_n(\infty)$ (which are associated with the same Schrödinger operator but on an graph $\Gamma$ with a Dirichlet boundary condition in each vertex $v \in \mathcal{V}$) we have $0 \leq 1-\frac{\lambda_n(\sigma)}{\lambda_N(\infty)}$ for every $1 \leq n \leq N$. Therefore,
	\begin{align}\label{agn:eigenvalueexpansionestimate}
	\sum_{n=1}^N \sum_{v \in \mathcal{V}} \vert f_n^{\sigma,V}(v) \vert^2 &\leq \sum_{n=1}^N \mathrm{e}^{1-\frac{\lambda^V_n(\sigma)}{\lambda^V_N(\infty)}}\sum_{v \in \mathcal{V}} \vert f_n^{\sigma,V}(v) \vert^2 = \mathrm{e} \sum_{n=1}^N \mathrm{e}^{-\frac{\lambda^V_n(\sigma)}{\lambda^V_N(\infty)}}\sum_{v \in \mathcal{V}} \vert f_n^{\sigma,V}(v) \vert^2 \nonumber \\&= \mathrm{e} \sum_{v \in \mathcal{V}} \sum_{n=1}^N \mathrm{e}^{-\frac{\lambda^V_n(\sigma)}{\lambda^V_N(\infty)}}\vert f_n^{\sigma,V}(v) \vert^2 \leq \mathrm{e} \sum_{v \in \mathcal{V}} p^{H_\sigma}_{e,e}\bigg(\frac{1}{\lambda^V_N(\infty)};v,v\bigg)\ ,
	\end{align}
	where we made use of the relation 
	\begin{equation}\label{RelationXXX}
	p_{e,e}^{H_\sigma}(t;v,v) = \sum_{n=1}^\infty \mathrm{e}^{-t\lambda^V_n(\sigma)} \vert f_n^{\sigma,V}(v)\vert^2\ .
	\end{equation}
	The estimate in~\eqref{agn:wurzelabschätzungheatkernel} yields
	\begin{equation*}\begin{split}
	\sum_{v \in \mathcal{V}} p^{H_\sigma}_{e,e}\bigg(\frac{1}{\lambda^V_N(\infty)};v,v\bigg) &\leq C_2(\Gamma,\hat{\sigma},\hat{V}_-) \sqrt{\lambda^V_N(\infty)} \\
	&\leq C_2(\Gamma,\hat{\sigma},\hat{V}_-) \sqrt{\lambda^{V = \hat{V}_+}_N(\infty)} 
	\end{split}
	\end{equation*}
	for $N \in \mathbb{N}$ large enough and $C_2(\Gamma,\hat{\sigma},\hat{V}_-) > 0$ some constant. We combine this with \eqref{agn:eigenvalueexpansionestimate} and Weyl's law which shows that $\lambda^{\hat{V}_+}_N(\infty) \sim \frac{\pi^2}{\mathcal{L}^2}N^2$ (note that this implies the dependence of the constant $C(\Gamma,\hat{\sigma},\hat{V}_-,\hat{V}_+)$ on $\hat{V}_+$). This gives the first statement.
	
	Finally, repeating the argument for some $x \in (0,l_e)$, we eventually arrive at 
	\begin{equation*}
	\sum_{n=1}^N \vert (f_n^{\sigma,V})_e(x)\vert^2  \leq C_3(\Gamma,\hat{\sigma},\hat{V}_-)  \cdot p^{H^{V=0}_0}_{e,e}\bigg(\frac{1}{\lambda^V_N(\infty)};x,x\bigg)
	\end{equation*}
	for some constant $C_3(\Gamma,\hat{\sigma},\hat{V})  > 0$. Using relation~\eqref{RelationXXX} and Lemma~\ref{Dominated} to bound the eigenfunctions of the operator $H^{V=0}_0$ as well as the heat-kernel asymptotics \cite[(8.14)]{BER}, we conclude 
	\begin{equation*}\begin{split}
	\sum_{n=1}^N \vert (f_n^{\sigma,V})_e(x)\vert^2  &\leq C_4(\Gamma,\hat{\sigma},\hat{V}_-)\sqrt{\lambda^V_N(\infty)} \\
	&\leq C_4(\Gamma,\hat{\sigma},\hat{V}_-)  \sqrt{\lambda^{V=\hat{V}_+}_N(\infty)} 
	\end{split}
	\end{equation*}
	for some constant $C_4(\Gamma,\hat{\sigma},\hat{V}_-) > 0$ independent of $x$. Weyl's law then gives the statement as above.
\end{proof}

	\subsection*{Acknowledgement}{PB was supported by the Deutsche Forschungsgemeinschaft DFG (Grant 397230547). We thank Delio Mugnolo (Hagen) for helpful remarks. JK also enjoyed interesting discussions with Matthias Täufer (Hagen) during the Mini-Workshop (A Geometric Fairytale full of Spectral Gaps and Random Fruit) at MFO. We thank Pavel Kurasov (Stockholm) for an interesting question during the recent QGraph meeting which lead to a generalization of Theorem~\ref{MainResult}.}
	
	\vspace*{0.5cm}
	
	{\small
		\bibliographystyle{amsalpha}
		\bibliography{Literature}}

\end{document}